\documentclass[journal]{IEEEtran}
\ifCLASSINFOpdf
\else
\fi
\usepackage{pgfplots}
\usepackage{amsmath}
\usepackage[utf8]{inputenc}
\usepackage[english]{babel}
\usepackage{graphicx}
\usepackage{tikz}
\usepackage{array}
\usepackage{booktabs}
\usepackage{caption}
\usepackage{tabu}
\usetikzlibrary{shapes,arrows,shadows}
\usepackage{epstopdf}
\usepackage{amsthm}

\newtheorem{lemma}{Lemma}

\usepackage{algorithm}
\usepackage{algpseudocode}
\usepackage{pifont}
\usepackage{subfig}
\usepackage{amssymb}

\theoremstyle{definition}
\newtheorem{definition}{Definition}

\newtheorem{property}{Property}

\makeatletter
\def\BState{\State\hskip-\ALG@thistlm}
\makeatother

\begin{document}
%
\title{Generalized Analysis of Convergence of Absolute Trust in Peer to Peer Networks }
%
%
%

\author{Sateesh Kumar Awasthi~and~Yatindra~Nath~Singh, \IEEEmembership{Senior~Member,~IEEE}
\thanks{Sateesh Kumar Awasthi  is with the Department
of Electrical Engineering, Indian Institute of Technology, Kanpur, India, e-mail: (sateesh@iitk.ac.in)}
\thanks{Yatindra~Nath~Singh  is with the Department
of Electrical Engineering, Indian Institute of Technology, Kanpur, India, e-mail: (ynsingh@iitk.ac.in)}}
\maketitle

\begin{abstract}
Open and anonymous nature of peer to peer networks provides an opportunity to malicious peers to behave unpredictably in the network. This leads the lack of trust among the peers. To control the behavior of peers in the network, reputation system can be used. In a reputation system, aggregation of trust is a primary issue. Algorithm for aggregation of trust should be designed such that, it can converge to a certain finite value. Absolute Trust is one of the algorithm, which is used for the aggregation of trust in peer to peer networks. In this letter, we present the generalized analysis of convergence of the Absolute Trust algorithm.
\end{abstract}


%
\IEEEpeerreviewmaketitle

\section{Introduction}
%
%
%
%
\IEEEPARstart{R}{eputation} systems have been proposed by many authors in the recent past \cite{abs}, \cite{eigen}, \cite{peertrust}, \cite{fuzzy}, \cite{power}, to prevent the attacks by rogue peers. Absolute Trust \cite{abs} is one such model. This model can characterize the past behavior of peers in the network. It can be implemented as a truly distributed system. In this model, peers evaluate each other locally, and the local trust for each other is aggregated in the whole network. The aggregated trust is called global trust. The global trust is evaluated recursively.
\par In recursive solution of any equation, error in each iteration must reduce and for large number of iteration it should tend to zero. This will guarantee the uniqueness of the solution. It was shown in \cite{abs} that if error in global trust is less compare to the actual solution, then it will converge to zero. But analysis for large error was not presented. In this letter, we will show that in any step, if error is very large compared to the actual solution then it will converge much faster in that step.   


\section{peer to peer model and absolute trust }
Let there be $N$ peers in a peer to peer network. In this network, peer $i$ can  be evaluated by peer $j$ based on service provided by peer $i$ in the past. Evaluated value $T_{ji}$ can be represented by a number from one  to ten. One is for worst service and ten is for best service. If there is no interaction between peers, $T_{ji}$ will be zero. $T_{ji}$ is called local trust of peer $i$ evaluated by  peer $j$. All  local trust values evaluated by various nodes can be aggregated in the whole network. Aggregated global trust of peer $i$ can be  given by \cite{abs}
\begin{equation}\label{eq1}
t_i=\Bigg[\Bigg(\frac{\sum_{j\in S_i}{T_{ji}t_j}}{\sum_{j\in S_i}{t_j}}\Bigg)^p . \Bigg(\frac{\sum_{j\in S_i}{t_{j}^2}}{\sum_{j\in S_i}{t_j}}\Bigg)^q\Bigg]^\frac{1}{(p+q)}.
 \end{equation}
 Here, $S_i$ is the set of peers evaluating the service of peer $i$, $t_j$ is global trust of peer $j$, and $p,q$ are suitably chosen constants. Equation \ref{eq1} can be rearranged as follows (see \cite{abs}).
 \[t_i=\Bigg[\bigg(\frac{\mathbf{e_iT^t t}}{\mathbf{e_iCt}}\bigg)^p\bigg(\frac{ \mathbf{e_iC. diag(t).t}}{\mathbf{e_iCt}}\bigg)^q \Bigg]^\frac{1}{p+q}\]

\[=\Bigg[\bigg(\frac{\mathbf{e_iT^t t}}{\mathbf{e_iCt}}\bigg)^\frac{1}{1+\alpha}\bigg(\frac{ \mathbf{e_iC. diag(t).t}}{\mathbf{e_iCt}}\bigg)^\frac{\alpha}{1+\alpha} \Bigg]\]

Here, $\mathbf{t}$ is global trust vector. Its $i^{th}$ entry is a global trust of peer $i$. $\mathbf{diag(t)}$ is a diagonal matrix with its $ii$ entry as $t_i$ and other entry as zero. $\mathbf{T}$ is trust matrix with its element $T_{ij}$ as a local trust of peer $j$ evaluated by peer $i$. $\mathbf{T^t}$ is transpose of matrix $\mathbf{T}$. $\mathbf{C}$ is incidence matrix corresponding to matrix $\mathbf{T^t}$ i.e. $C_{ij}=1$ if $T_{ji}>0$ otherwise $C_{ij}=0$. $\mathbf{e_i}$ is the row vector with its $i^{th}$ entry as '1' and all the other entry as zero. Here, $\alpha =q/p.$


\section{Analysis of convergence of Absolute trust}
\subsection{Center Point of the Matrix}
\begin{definition}\label{def1}
Center point of non-negative matrix $\mathbf{T^t}$ can be defined as the column vector $\mathbf{t}$. Where, its $i^{th}$ element $t_i$ will be $(\mathbf{e_iT^t t}/\mathbf{e_iCt})$. 
\end{definition}
\begin{definition}\label{def2}
A non-negative matrix $\mathbf{M}$ is said to be mutually exclusive with a non-negative matrix $\mathbf{N}$, if $M_{ij}>0$ implies $N_{ij}=0$. It also implies that when $N_{ij}>0$ then $M_{ij}=0$
\end{definition}
\begin{lemma}\label{lemma1}
Center point of any non-negative, irreducible matrix $\mathbf{T^t}$ is unique and can be calculated by an iterative function 
\[\mathbf{t^k}=\phi_1(\mathbf{t^{k-1}})=[\mathbf{diag}(d_1, d_2....d_N)]^{-1}.\mathbf{T^t. t^{k-1}}\] where $d_i=\mathbf{e_iCt}$
\end{lemma}
\begin{proof}
$i^{th}$ element of iterative function $\phi_1(\mathbf{t^{k-1}})$ is
\begin{equation}\label{equ2}
t_i^{k}=\frac{(\mathbf{e_iT^t .t^{k-1}})}{(\mathbf{e_iC.t^{k-1}})}
\end{equation}

Let $t_i^k$ and $t_i^{k-1}$ are, far from actual solution $t_i$ by  $\delta t_i^k$ and $\delta t_i^{k-1}$ respectively, then  

\[t_i + \delta t_i^{k}=  \Bigg[\frac{\mathbf{\big(e_i T^t.(t + \delta t^{k-1})\big)}}{\mathbf{\big(e_i C.(t + \delta t^{k-1})\big)}}\Bigg]\]

\[=\Bigg[\frac{(\mathbf{e_iT^t .t})}{(\mathbf{e_iC.t})}\Bigg]\Bigg[\frac{(1+\frac{\mathbf{e_i.T^t.\delta t^{k-1}}}{\mathbf{e_i.T^t.t}})}{(1+\mathbf{\frac{e_i.C.\delta t^{k-1}}{e_i.C.t})}}\Bigg].\]
From definition \ref{def1}, $(\mathbf{e_iT^t t}/\mathbf{e_iCt})$ is $i^{th}$ element of center point of matrix $\mathbf{T^t}$. So we can write,
\[t_i + \delta t_i^{k}=t_i.\Bigg[\frac{(1+\frac{\mathbf{e_i.T^t.\delta t^{k-1}}}{\mathbf{e_i.T^t.t}})}{(1+\mathbf{\frac{e_i.C.\delta t^{k-1}}{e_i.C.t})}}\Bigg]\]

\[\delta t_i^{k}=t_i.\Bigg[\frac{(1+\frac{\mathbf{e_i.T^t.\delta t^{k-1}}}{\mathbf{e_i.T^t.t}})}{(1+\mathbf{\frac{e_i.C.\delta t^{k-1}}{e_i.C.t})}}-1\Bigg]\]

\[\delta t_i^{k}=\frac{t_i}{(1+\mathbf{\frac{e_i.C.\delta t^{k-1}}{e_i.C.t})}}.\Bigg[\frac{\mathbf{e_i.T^t.\delta t^{k-1}}}{\mathbf{e_i.T^t.t}}-\mathbf{\frac{e_i.C.\delta t^{k-1}}{e_i.C.t}}\Bigg]\]

\[=\frac{1}{(1+\mathbf{\frac{e_i.C.\delta t^{k-1}}{e_i.C.t})}}.\Bigg[\frac{t_i\mathbf{e_i.T^t}}{\mathbf{e_i.T^t.t}}-\frac{t_i\mathbf{e_i.C.}}{\mathbf{e_i.C.t}}\Bigg].\mathbf{\delta t^{k-1}}\]

\[=\frac{1}{(1+\mathbf{\frac{e_i.C.\delta t^{k-1}}{e_i.C.t})}}.\bigg[\mathbf{A_i}-\mathbf{B_i}\bigg].\mathbf{\delta t^{k-1}}\]

\[=f_i(\mathbf{\delta t^{k-1}}).\bigg[\mathbf{A_i}-\mathbf{B_i}\bigg].\mathbf{\delta t^{k-1}}\]
where $\mathbf{A_i}$ and $\mathbf{B_i}$ are $i^{th}$ row of $NXN$ matrices $\mathbf{A}$ and $\mathbf{B}$ respectively. It can be  observed easily that \[\mathbf{A.t}=\mathbf{t}\] and \[\mathbf{B.t}=\mathbf{t}.\]
 Matrices $ \mathbf{A,B} $ have non zero elements  at same position as matrix $\mathbf{T^t}$, hence $ \mathbf{A,B} $  are also irreducible. Therefore spectral radius of $\mathbf{A , B}$  will be '1'(see\cite{nonnegative}).
 \par Now \[f_i(\mathbf{\delta t^{k-1}})=\frac{1}{(1+\mathbf{\frac{e_i.C.\delta t^{k-1}}{e_i.C.t})}}\]
 If $\mathbf{\delta t^{k-1}} << \mathbf{t}$ then $f_i(\mathbf{\delta t^{k-1}})\approx 1$ \par Hence, 
\[\mathbf{\delta t^k}=\big[\mathbf{A}-\mathbf{B}\big].\mathbf{\delta t^{k-1}}\]
\[lim_{k\to\infty}\mathbf{\delta t^k} =lim_{k\to\infty}[\mathbf{A-B}]^k\mathbf{\delta t^0}=\mathbf{0}\] (see Theorem 1 in \cite{abs}) here $\mathbf{\delta t^0}$ is initial error in $\mathbf{t}$ \par Now, for the case when $\mathbf{\delta t^{k-1}} > \mathbf{t}$ then $f_i(\mathbf{\delta t^{k-1}})< 1$, in each step, $\delta t_i^k$ will decrease more rapidly. If $\mathbf{\delta t^{k-1}} \approx \mathbf{t}$ then $f_i(\mathbf{\delta t^{k-1}}) \approx 1/2$, in this case each $\delta t_i^k$ will reduced to half in $k^{th}$ step.\par Hence we can conclude that center point of any non negative, irreducible matrix can be calculated by above iterative function. Error in each step will depend upon the error in past step. Error will reduce very fast, if it is far from actual solution in an step and after large iterations it will go to zero.
\end{proof}

\begin{lemma}\label{lemma2}
If  vector $\mathbf{t}$ is the center point of matrix $\mathbf{C.diag(t)}$, then center point will lie on the vector $\mathbf{e}$. It can be calculated by iterative function 
\[\mathbf{t^k}=\phi_2(\mathbf{t^{k-1}})=[\mathbf{diag}(d_1, d_2....d_N)]^{-1}.\mathbf{C.diag(t^{k-1}).t^{k-1}}\] where $\mathbf{e}$ is a vector with each element as '1' and $d_i=\mathbf{e_iCt}.$
\end{lemma}
\begin{proof}
Iterative function can be written as 
\[\mathbf{t^k=M(k-1).t^{k-1}}\]
\[=\mathbf{M(k-1).M(k-2)....M(0).t^0}\]
where $i^{th}$ row of matrix $\mathbf{M(k)}$ is  $\big(\mathbf{\frac{e_iC.diag(t^k)}{e_iC.t^k}}\big)$. We can easily prove that for matrix $\mathbf{M(k)}$, sum of its each row is one. Hence it has '1' as an  eigen value. The corresponding eigen vector will be $\mathbf{e}$. For positive initial guess of $\mathbf{t^0}$, all matrices $\mathbf{M(k)}$ will be non-negative and   irreducible. So we can conclude that spectral radius of all $\mathbf{M(k)}$ is '1'. Hence (see Lemma 2 in \cite{abs})
\[lim_{k\to\infty}\mathbf{t^k} =lim_{k\to\infty}\mathbf{M(k-1).M(k-2)....M(0).t^0=e}\]
\[lim_{k\to\infty}\mathbf{t^k} =lim_{k\to\infty}\phi_2(\mathbf{t^{k-1}})=\mathbf{e}\]
\end{proof}
\subsection{Properties of Center Point}
\begin{property}\label{prop1}
If center point of matrix $\mathbf{M}$ is $\mathbf{t}$ then center point of matrix $k\mathbf{M}$ will be $k\mathbf{t}$. Where $k$ is any arbitrary scalar.

\end{property}
\begin{proof}
Let $t_i$ is $i^{th}$ element of $\mathbf{t}$ then
\[t_i =\mathbf{\frac{e_i.M.t}{e_i.C.t}}\]
or
\[kt_i =\frac{\mathbf{e_i}.k\mathbf{M}.\mathbf{t}}{\mathbf{e_i.C.t}}\]
or
\[(kt_i) =\frac{\mathbf{e_i}.(k\mathbf{M}).(k\mathbf{t})}{\mathbf{e_i.C}.(k\mathbf{t})}\] hence $k\mathbf{t}$ is center point of matrix $k\mathbf{M}$
\end{proof}
\begin{property}\label{prop2}
If all entries of the $NXN$ matrix $\mathbf{M}$ are positive($M_{ij}>0 \hspace{2mm}\forall i,j$ ) then center point will lie on the principle eigen vector of matrix.
\end{property}
\begin{proof}
If all the entry of matrix $\mathbf{M}$ are positive then
\[\mathbf{e_i.C.t} = \sum_{j=1}^{N}t_j=\lambda \hspace{5mm}  \forall i\] Here $\lambda$ is a constant. Further,
\[t_i =\mathbf{\frac{e_i.M.t}{e_i.C.t}}=\frac{\mathbf{e_i.M.t}}{\lambda}.\]
Hence $\mathbf{Mt}=\lambda\mathbf{t}$. Hence $\lambda$ will be spectral radius and $\mathbf{t}$ will be principle eigen vector of matrix $\mathbf{M}$ (see\cite{nonnegative})
\end{proof}

\begin{property}\label{prop3}
If center point of $m$ non-negative, irreducible and \textbf{mutually exclusive} matrices  $\mathbf{M_1,M_2....M_m}$  are same, then the center point of their sum will also be the  same.
\end{property}
\begin{proof}
Let the incidence matrix corresponding to matrices $\mathbf{M_1,M_2....M_m}$ are $\mathbf{C_1,C_2....C_m}$ respectively. Now if $\mathbf{M=M_1+M_2+....M_m}$ then incidence matrix corresponding to matrix $\mathbf{M}$ will be $\mathbf{C=C_1+C_2+....C_m}$ because matrices $\mathbf{M_1,M_2....M_m}$ are mutually exclusive. Let $i^{th}$ element of center point of  matrices $\mathbf{M_1,M_2....M_m}$ be $t_i$ then
\[t_i =\mathbf{\frac{e_i.M_j.t}{e_i.C_j.t}} \hspace{5mm} \forall j\]
or
\[ t_i\mathbf{e_i}.\mathbf{C_j.t}=\mathbf{e_i.M_j.t}\]
taking summation on both side w.r.t. $j$
\[\sum_{j=1}^{m}t_i\mathbf{e_i}.\mathbf{C_j.t}=\sum_{j=1}^{m} \mathbf{e_i.M_j.t}\]

\[t_i\mathbf{e_i}.\big(\sum_{j=1}^{m}\mathbf{C_j}\big).\mathbf{t}=\mathbf{e_i.\big(\sum_{j=1}^{m}M_j\big).t}\]

\[t_i\mathbf{e_i}.\mathbf{C.t}=\mathbf{e_i.M.t}\]
or
\[ t_i =\mathbf{\frac{e_i.M.t}{e_i.C.t}}\]
hence $\mathbf{t}$ is also the center point of matrix $\mathbf{M}.$ 

\end{proof}

\subsection{Convergence of Absolute Trust for Large Error in Initial Guess}
It was stated in \cite{abs} that global trust in equation \ref{eq1} can be calculated by iterative function 
\[\mathbf{t^{k}}  = \mathbf{\phi(t^{k-1})}=[\mathbf{diag}(d_1,d_2....d_N) \mathbf{T^t. t^{k-1}}]^\frac{1}{1+\alpha}\]
where $d_i = \bigg[\frac{{\big(\mathbf{e_i C. diag(t^{k-1}) .t^{k-1}}}\big)^\alpha} {{\big(\mathbf{e_i C t^{k-1}}}\big)^{(1+\alpha)}} \bigg].$ Proof was derived only for the small error that is, if initial guess of $\mathbf{t}$ is very close to the actual solution. However global trust can be calculated for any positive initial guess. In this subsection we will show that it will converge faster in any step if error is large compared to  $\mathbf{t}$.
\par
Let $t_i^k$ and $t_i^{k-1}$ are, far from actual solution $t_i$ by  $\delta t_i^k$ and $\delta t_i^{k-1}$ respectively, then  
\begin{multline*}
t_i + \delta t_i^{k}=  \Bigg[\frac{\mathbf{\big(e_i T^t.(t + \delta t^{k-1})\big)}}{\mathbf{\big(e_i C (t + \delta t^{k-1})\big)}}\Bigg]^\frac{1}{1+\alpha}.\\
\Bigg[\frac{\mathbf{\big(e_i C. diag(t + \delta t^{k-1}).(t + \delta t^{k-1})\big)}}{\mathbf{\big(e_i C (t + \delta t^{k-1})\big)}}\Bigg]^\frac{\alpha}{1+\alpha}
\end{multline*}

If error $\mathbf{\delta t^{k-1}> t}$ then we can approximate $\mathbf{t+\delta t^{k-1} \approx \delta t^{k-1} }$ hence
\[\delta t_i^{k}=  \Bigg[\frac{\mathbf{\big(e_i T^t.\delta t^{k-1}\big)}}{\mathbf{\big(e_i C.\delta t^{k-1}\big)}}\Bigg]^\frac{1}{1+\alpha}.
\Bigg[\frac{\mathbf{\big(e_i C. diag(\delta t^{k-1}).\delta t^{k-1}\big)}}{\mathbf{\big(e_i C.\delta t^{k-1}\big)}}\Bigg]^\frac{\alpha}{1+\alpha}\]
Using Young's Inequality \cite{inequality}, i.e.
\[c.d \leq \frac{1}{1+\alpha}c^{1+\alpha}+\frac{\alpha}{1+\alpha}d^{(1+\alpha)/\alpha};\]
and taking  
\[c=\Bigg[\frac{\mathbf{\big(e_i T^t.\delta t^{k-1}\big)}}{\mathbf{\big(e_i C.\delta t^{k-1}\big)}}\Bigg]^\frac{1}{1+\alpha}\] and
 \[d=\Bigg[\frac{\mathbf{\big(e_i C. diag(\delta t^{k-1}).\delta t^{k-1}\big)}}{\mathbf{\big(e_i C.\delta t^{k-1}\big)}}\Bigg]^\frac{\alpha}{1+\alpha},\]
We can write 
\begin{multline*}
\delta t_i^{k} \leq  \frac{1}{1+\alpha}\Bigg[\frac{\mathbf{\big(e_i T^t.\delta t^{k-1}\big)}}{\mathbf{\big(e_i C.\delta t^{k-1}\big)}}\Bigg]+\\
\frac{\alpha}{1+\alpha}\Bigg[\frac{\mathbf{\big(e_i C. diag(\delta t^{k-1}).\delta t^{k-1}\big)}}{\mathbf{\big(e_i C.\delta t^{k-1}\big)}}\Bigg].
\end{multline*}

\[\mathbf{\delta t^{k}} \leq \frac{1}{1+\alpha}\phi_1(\mathbf{\delta t^{k-1}}) + \frac{\alpha}{1+\alpha}\phi_2(\mathbf{\delta t^{k-1}}).\]
Here $\phi_1(.)$ and $\phi_2(.)$ are as defined in Lemma \ref{lemma1} and in Lemma \ref{lemma2} respectively.

\begin{equation}\label{equ3}
\mathbf{\delta t^{k}} \leq \frac{1}{1+\alpha}\mathbf{M^{'}_1.\delta t^{k-1}} + \frac{\alpha}{1+\alpha}\mathbf{M^{'}_2.\delta t^{k-1}}
\end{equation}

It is convex combination of iterative function $\phi_1$ and $\phi_2$. $i^{th}$ row of matrix $\mathbf{M^{'}_2}$ is $\bigg[\frac{\mathbf{\big(e_i C. diag(\delta t^{k-1})\big)}}{\mathbf{\big(e_i C.\delta t^{k-1}\big)}}\bigg]$ and sum of each row is one. Hence $\infty-$norm of matrix $\mathbf{M^{'}_2}$ is $\mathbf{|{M^{'}_2}|_\infty}=1$ using the property of norm \cite{matrix}
\begin{equation}\label{equ4}
|\mathbf{M^{'}_2.\delta t^{k-1}|_\infty} \leq |\mathbf{M^{'}_2|_\infty}.|\mathbf{\delta t^{k-1}|_\infty} =|\mathbf{\delta t^{k-1}|_\infty}
\end{equation}
Function $\phi_1$ is converging function toward center point. It is shown in Lemma \ref{lemma1} that in any step, if $\mathbf{t^k}$ is  very far from the center point  then it will tend toward the center point very rapidly. Hence for $\mathbf{\delta t^{k-1} > t}$
\[\phi_1(\mathbf{\delta t^{k-1}})<\mathbf{\delta t^{k-1}}\]
or
\begin{equation}\label{equ5}
|\mathbf{M^{'}_1.\delta t^{k-1}|_\infty} < |\mathbf{\delta t^{k-1}|_\infty}
\end{equation} 
Now taking $\infty$-norm on both side of equation \ref{equ3}  
\[|\mathbf{\delta t^{k}}|_\infty \leq |\frac{1}{1+\alpha}\mathbf{M^{'}_1.\delta t^{k-1}} + \frac{\alpha}{1+\alpha}\mathbf{M^{'}_2.\delta t^{k-1}}|_\infty\]

\[\leq \frac{1}{1+\alpha}|\mathbf{M^{'}_1.\delta t^{k-1}}|_\infty + \frac{\alpha}{1+\alpha}|\mathbf{M^{'}_2.\delta t^{k-1}}|_\infty\]
Using equation \ref{equ4} and \ref{equ5},
 
 \[|\mathbf{\delta t^{k}}|_\infty < \frac{1}{1+\alpha}|\mathbf{\delta t^{k-1}}|_\infty + \frac{\alpha}{1+\alpha}|\mathbf{\delta t^{k-1}}|_\infty=|\mathbf{\delta t^{k-1}}|_\infty\]
 Hence 
 \[|\mathbf{\delta t^k}|_\infty < |\mathbf{\delta t^{k-1}}|_\infty\]
 Therefore error in every step will decrease. In any step, it will decrease faster if $\mathbf{t^{k}}$ is very far from actual solution. Speed of convergence depends upon the  $\alpha$. For lower $\alpha$, impact of $\phi_2$ will be lower and $\phi_1$ will dominate the speed of convergence. Hence for smaller $\alpha$ speed of convergence will be high.

\section{Numerical Results}
In order to verify what has been discussed in earlier sections, we have taken the values of  the trust matrix $\mathbf{T}$ as  
$$\quad
\begin{bmatrix} 
0 & 5&6&6 \\
8&0&5&5\\
5&6&0&2\\
0&4&0&0
 
\end{bmatrix}
\quad.$$ 
The value of center point and the global trust in each iteration is calculated and shown in Table \ref{table1}, \ref{table2}, \ref{table3}, \ref{table4} and \ref{table5}.
\subsection{Convergence of the Center Point}
The convergence of center point of matrix $\mathbf{T^t}$ is shown in Table \ref{table1} and \ref{table2}. In Table \ref{table1} initial guess $\mathbf{t^0}=$[1 2 3 4]$^t$, it is close to the center point and in Table \ref{table2} initial guess $\mathbf{t^0}=$[100 300 200 100]$^t$, which is significantly far from center point. But we can see in both the cases that it will converge in seven iterations . In the latter case, error is very large in $0^{th}$ step and it becomes less then $\mathbf{t}$ with in one step.
\subsection{Convergence of the Global Trust}
Impact of initial guess and parameter $\alpha$ on the convergence of global trust is shown in Table \ref{table3}, \ref{table4} and \ref{table5}.  In Table \ref{table3} and \ref{table4} $\alpha$ is taken as 1/3 but  initial guess is different. Again we can see that for large initial guess of global trust it takes only one more iteration to converge to final value. It converge very fast when error($\mathbf{\delta t}$) is very large compare to global trust($\mathbf{t}$)
\par In Table \ref{table3} and \ref{table5} initial guess is taken same but $\alpha$ is different and we can see that for $\alpha=1/6$ it converges only in eight iterations.
\begin{table}
\begin{center}
\caption{Center point in each iteration, when initial guess is close to center point}\label{table1}
\begin{tabu} to 0.4\textwidth { | X[c] | X[c] | X[c] | X[c] | X[c]| X[c] | } 
 \hline
 $i$& 1& 2& 3& 4\\[1ex] 
 \hline
 $\mathbf{ t^0}$& 1.0000 &   2.0000 &   3.0000 &   4.0000\\
    \hline
$\mathbf{ t^1}$& 6.2000 &   4.8750 &   5.3333   & 3.6667\\
    \hline
 $\mathbf{ t^2}$&6.4327  &  5.1096  &  5.5598&    4.4027\\
    \hline
 $\mathbf{ t^3}$&6.4367   & 5.0706   & 5.5573 &   4.4008\\
    \hline
 $\mathbf{ t^4}$&6.4313    &5.0705    &5.5594  &  4.4002\\
    \hline
 $\mathbf{ t^5}$&6.4310    &5.0707    &5.5592   & 4.3994\\
    \hline
 $\mathbf{ t^6}$&6.4311   & 5.0708    &5.5591    &4.3994\\
    \hline
 $\mathbf{ t^7}$&6.4311   & 5.0708    &5.5591    &4.3994\\

  \hline
\end{tabu}
\end{center}
\end{table}

\begin{table}
\begin{center}
\caption{Center point in each iteration, when initial guess is very far from center point }\label{table2}
\begin{tabu} to 0.4\textwidth { | X[c] | X[c] | X[c] | X[c] | X[c]| X[c] | } 
 \hline
 $i$& 1& 2& 3& 4\\[1ex] 
 \hline
$\mathbf{ t^0}$& 100.0000&  300.0000&  200.0000 & 100.0000\\ 
 \hline
$\mathbf{ t^1}$&    6.8000&    5.2500&    5.2500 &   4.1667\\ 
 \hline
$\mathbf{ t^2}$&    6.5000 &   5.0668 &   5.5643  &  4.4827\\ 
 \hline
$\mathbf{ t^3}$&    6.4298  &  5.0654  &  5.5620   & 4.4050\\ 
 \hline
$\mathbf{ t^4}$&    6.4299   & 5.0706   & 5.5593    &4.3987\\ 
 \hline
$\mathbf{ t^5}$&    6.4310    &5.0708    &5.5591 &   4.3993\\ 
 \hline
$\mathbf{ t^6}$&    6.4311 &   5.0708  &  5.5591  &  4.3994\\ 
 \hline
$\mathbf{ t^7}$&    6.4311  &  5.0708   & 5.5591   & 4.3994\\

  \hline
\end{tabu}
\end{center}
\end{table}

\begin{table}
\begin{center}
\caption{Global trust in each iteration, when initial guess is close to global trust ($\alpha =1/3$) }\label{table3}
\begin{tabu} to 0.4\textwidth { | X[c] | X[c] | X[c] | X[c] | X[c]| X[c] | } 
 \hline
 $i$& 1& 2& 3& 4\\[1ex] 
 \hline
 
$\mathbf{ t^0}$&1.0000 &   2.0000 &   3.0000 &   4.0000\\ 
 \hline
$\mathbf{ t^1}$&4.9893 &   4.4051  &  3.9876  &  3.2749\\
 \hline
$\mathbf{ t^2}$&5.8801  &  4.8299   & 5.3148   & 4.4838\\ 
 \hline
$\mathbf{ t^3}$&6.0621&    5.1110&    5.5131&   4.6039\\ 
 \hline
$\mathbf{ t^4}$&6.1418 &   5.1543 &   5.5636 &   4.6490\\ 
 \hline
$\mathbf{ t^5}$&6.1550   & 5.1683   & 5.5802   & 4.6631\\
 \hline
$\mathbf{ t^6}$&6.1593  &  5.1717   & 5.5834   & 4.6657\\ 
 \hline
$\mathbf{ t^7}$& 6.1603  &  5.1725   & 5.5844   & 4.6665\\ 
 \hline
$\mathbf{ t^8}$& 6.1605   & 5.1727  &  5.5846    &4.6667\\ 
 \hline
$\mathbf{ t^9}$& 6.1606    &5.1728   & 5.5847   & 4.6667\\ 
 \hline
$\mathbf{ t^{10}}$& 6.1606 & 5.1728    &5.5847    &4.6667\\
  \hline
\end{tabu}
\end{center}
\end{table}

\begin{table}
\begin{center}
\caption{Global trust in each iteration, when initial guess is very far from global trust ($\alpha =1/3$)}\label{table4}
\begin{tabu} to 0.4\textwidth { | X[c] | X[c] | X[c] | X[c] | X[c]| X[c] | } 
 \hline
 $i$& 1& 2& 3& 4\\[1ex] 
 \hline
$\mathbf{ t^0}$&100.0000&  300.0000&  200.0000 & 100.0000\\ 
 \hline
$\mathbf{ t^1}$&16.9093  & 12.1379  & 13.7913   &11.3982\\ 
 \hline
$\mathbf{ t^2}$& 7.6469&    6.5685&    7.1369    &5.9630\\ 
 \hline
$\mathbf{ t^3}$&  6.5419 &   5.4824 &   5.9029    &4.9291\\ 
 \hline
$\mathbf{ t^4}$&  6.2499  &  5.2477  &  5.6685  &  4.7377\\ 
 \hline
$\mathbf{ t^5}$&  6.1829   & 5.1918   & 5.6048   & 4.6834\\ 
 \hline
$\mathbf{\ t^6}$&  6.1662    &5.1775    &5.5897    &4.6710\\ 
 \hline
$\mathbf{ t^7}$&  6.1620&    5.1740  &  5.5859   & 4.6678\\ 
 \hline
$\mathbf{ t^8}$&  6.1610 &   5.1731   & 5.5850    &4.6670\\ 
 \hline
$\mathbf{ t^9}$&  6.1607  &  5.1729    &5.5847   & 4.6668\\ 
 \hline
$\mathbf{ t^{10}}$& 6.1606 &  5.1728   & 5.5847  & 4.6667\\ 
 \hline
$\mathbf{ t^{11}}$&  6.1606 & 5.1728    &5.5847   &4.6667\\

  \hline
\end{tabu}
\end{center}
\end{table}

\begin{table}
\begin{center}
\caption{Global trust in each iteration for lesser value of $\alpha$( $\alpha =1/6$)}\label{table5}
\begin{tabu} to 0.4\textwidth { | X[c] | X[c] | X[c] | X[c] | X[c]| X[c] | } 
 \hline
 $i$& 1& 2& 3& 4\\[1ex] 
 \hline
$\mathbf{ t^0}$&1.0000&    2.0000 &   3.0000 &   4.0000\\ 
 \hline
$\mathbf{ t^1}$&    5.4761&    4.6006 &   4.5168&  3.4374\\ 
 \hline
$\mathbf{ t^2}$&    6.1901 &   5.0137  &  5.4742 & 4.5092\\ 
 \hline
$\mathbf{ t^3}$&    6.2506  &  5.1176   & 5.5682&  4.5424\\ 
 \hline
$\mathbf{ t^4}$&    6.2693   & 5.1288    &5.5767 & 4.5486\\ 
 \hline
$\mathbf{ t^5}$&    6.2714    &5.1304   & 5.5790  &4.5510\\ 
 \hline
$\mathbf{ t^6}$&    6.2716  &  5.1307    &5.5793&  4.5511\\ 
 \hline
$\mathbf{ t^7}$&    6.2717   & 5.1307   & 5.5793 & 4.5512\\ 
 \hline
$\mathbf{ t^8}$&    6.2717    &5.1307    &5.5793 & 4.5512\\

  \hline
\end{tabu}
\end{center}
\end{table}

\section{Conclusion}
In this letter, we analyzed the convergence of global trust given by equation \ref{eq1}. We have shown that in recursive calculation of global trust, error will decrease even if initial guess is very far from the actual solution. In any step, convergence is faster if error is larger. We have shown that speed of convergence depends upon the value of $\alpha$. For smaller $\alpha$ it will converge more faster.


%





\ifCLASSOPTIONcaptionsoff
  \newpage
\fi

\end{document}